\documentclass{ifacconf}
\usepackage{verbatim}
\usepackage{amsmath}
\usepackage{amssymb}   % includes mathbb
\usepackage{multirow}
\usepackage{graphicx}      
\usepackage{natbib}        % required for bibliography
\usepackage{mathrsfs}
\usepackage{amssymb}   % defines \mathbb
\usepackage{mathrsfs}  % defines \mathscr
\let\theoremstyle\relax
\usepackage{amsthm}
\usepackage{amsmath}   % for advanced math typesetting
\usepackage{dsfont}
\usepackage{multirow}
\usepackage{caption}
\usepackage{booktabs} % optional, improves table aesthetics
\usepackage[table]{xcolor}  % loads xcolor with table support
\usepackage{colortbl}       % provides \cellcolor

\newtheorem{theorem}{Theorem}

\newtheorem{definition}{Definition}
\newtheorem{remark}{Remark}
\newcommand{\R}{\mathbb{R}}
\usepackage{graphicx}      % include this line if your document contains figures
\usepackage{natbib}        % required for bibliography
\begin{document}
\begin{frontmatter}

\title{A Physics-Informed Scenario Approach with Data Mitigation for Safety Verification of Nonlinear Systems
}
\author[First]{Ali Aminzadeh} 
\author[Second]{MohammadHossein Ashoori} 
\author[Second]{Amy Nejati}
\author[Second]{Abolfazl Lavaei}

\address[First]{Faculty of Engineering and Natural Sciences, 
   Tampere University, Finland 
   (e-mail: ali.aminzadeh@tuni.fi).}

\address[Second]{School of Computing, 
   Newcastle University, United Kingdom 
   (e-mail: m.ashoori2@newcastle.ac.uk, amy.nejati@newcastle.ac.uk, abolfazl.lavaei@newcastle.ac.uk).}

\begin{abstract}
This paper develops a \emph{physics-informed} scenario approach for safety verification of nonlinear systems using barrier certificates (BCs) to ensure that system trajectories remain within safe regions over an infinite time horizon. Designing BCs often relies on an accurate dynamics model; however, such models are often imprecise due to the model complexity involved, particularly when dealing with highly nonlinear systems. In such cases, while scenario approaches effectively address the safety problem using collected data to construct a guaranteed BC for the unknown dynamical system, they often require solving an optimization problem with substantial amounts of data. To address this, we propose a {physics-informed} scenario approach that selects data samples such that the outputs of the physics-based model and the observed data are sufficiently close. This approach guides the scenario optimization process to eliminate redundant samples and potentially reduce the required dataset size. We validate our approach through three case studies, showcasing its practical application in reducing the required data.
\end{abstract}


\end{frontmatter}

\section{Introduction}
	Safety-critical systems are integral to the functioning of modern society, playing a pivotal role in a wide range of applications, including autonomous vehicles, aerospace, transportation systems, and healthcare. Ensuring the safety of these systems through {formal verification} is crucial, as their failure can result in significant consequences, such as loss of life or substantial financial losses~\citep{mcgregor2017analysis}. 
	In recent years, significant interest has emerged in verifying the safety of dynamical systems using {barrier certificates (BCs)}, a concept initially introduced by~\citet{prajna2004safety}. Specifically, BCs operate by enforcing a series of inequalities on the system's state and dynamics, analogous to a Lyapunov function. Through designing an appropriate level set, BCs ensure that system trajectories remain within safe regions, starting from a given initial set, thereby offering a formal safety guarantee. Barrier certificates have widely been used in the relevant literature for ensuring system safety and synthesizing controllers across various complex systems \citep[see \emph{e.g.,}][]{wieland2007constructive, ames2019control, zaker2024compositional, lavaeiTAC2024,nejati2024context,lavaei2022compositional}.

While effective, BCs typically require {accurate system models}. However, in practice especially for nonlinear systems, such models are often imprecise, creating a critical gap that limits formal safety guarantees when the dynamics are not fully reliable. To address this challenge, increasing attention has turned to the {scenario approach} \citep{calafiore2006scenario,campi2009scenario,margellos2014road}, which relies on sampled data rather than exact models. This data-driven framework has been successfully applied to safety verification and controller synthesis in recent years \citep[see \emph{e.g.,}][]{makdesi2021efficient,coppola2022data,nejati2023formal,aminzadeh2024compositional,banse2024data,samari2026data,zaker2025data}.
	
	While the scenario approach holds great promise for providing formal safety guarantees in dynamical systems, it relies on data being {independent and identically distributed (i.i.d.)}. Consequently, only a single input-output data pair can be extracted from each trajectory~\citep{calafiore2006scenario}, necessitating the collection of \emph{multiple independent trajectories}. 
	This requirement limits the applicability of the scenario approach to systems that are resource-constrained or challenging to simulate extensively, as generating a large number of independent trajectories may be prohibitively expensive.
	
	\textbf{Contributions.} Motivated by this central challenge, this work introduces a \emph{physics-informed} scenario approach to ensure system safety while reducing the required amount of data. In particular, while the physics-based model is not entirely precise, it contains valuable insights derived from fundamental physical laws and can effectively reduce the data requirements for solving a scenario optimization program. Our approach focuses on those data samples where the outputs of the physics-based model and the observed data are sufficiently close, within a predefined threshold, demonstrating the similarity in the behavior of these two systems. This enables the physics-based model to guide the scenario optimization process, eliminating redundant data and thus potentially reducing the dataset size required for safety analysis.
	
	\textbf{Related literature on physics-informed techniques.} Physics-informed data-driven methods have recently gained significant
	interest in the realm of formal verification and control synthesis. In this regard, 	\citet{zhang2022physics} introduce a physics-informed data-driven approach for identifying governing equations from data and solving them to derive spatio-temporal responses. Additionally, \citet{niknejad2023physics} propose a physics-informed data-driven approach for control design in discrete-time linear systems. In the realm of physics-informed neural networks, \citet{huang2022applications} study their role in addressing challenges faced by deep learning applications in power systems, including issues related to data quality, infeasible solutions, and limited generalizability. Recently, \citet{liu2025formally} explore the use of physics-informed learning and formal verification for neural network-based control Lyapunov functions, which are crucial for stabilizing nonlinear systems.  While these studies show promise, none have explored a physics-informed \emph{scenario approach} for safety verification, which is studied in this paper.
	
	\section{Problem Description} \label{prelim}
	\subsection{Notation}
	We use $\mathbb{R}$, $\mathbb{R}_{>0}$, $\mathbb{R}_{\geq 0}$, $\mathbb{N}_0$, and $\mathbb{N}$ to denote the sets of real numbers, positive  and non-negative real numbers, non-negative and positive integers, respectively. The notation $\mathsf d = [{\mathsf d}_1; \ldots; {\mathsf d}_n]$ is employed to represent a vector of $n$ decision variables. The Euclidean norm of ${x} \in \mathbb{R}^n$ is denoted by $\Vert{x}\Vert$.
    \begin{comment}
        \end{comment}

	\subsection{Discrete-Time Nonlinear Systems}
	In this work, we focus on discrete-time nonlinear systems (dt-NS), which are characterized as 
	\begin{equation}\label{Eq_1a}
		\Lambda\!:x(k+1) = {f(x(k))}, \quad k\in\mathbb N_0,
	\end{equation}
	where $x \in {X}$ is the system's state, with ${X} \subset \mathbb R^n$ being its  state set, and $f\!\!: X \rightarrow X$ is the transition map which is assumed to be \emph{unknown} in our setting. The sequence \( x_{x_0}\!\!: \mathbb{N}_0 \rightarrow X \) that satisfies~\eqref{Eq_1a} for any initial state \( x_0 \in X \) is referred to as the {state trajectory} of \( \Lambda \), originating from the initial state \( x_0 \). We use the tuple $(X,f)$ to refer to the dt-NS in~\eqref{Eq_1a}.
	
	In the following subsection, we present the concept of barrier certificates for the dt-NS described in \eqref{Eq_1a}, which can be used to provide a safety assurance for the system.
	
	\subsection{Barrier Certificates} \label{barrier}
	
	\begin{definition} \label{BC}
		Consider a dt-NS $\Lambda = (X,f)$, with $X_0$ and $X_u$ representing its initial and unsafe sets, respectively. A function $\mathcal{B}\!: X \to \mathbb{R}$ is considered a barrier certificate (BC) for $\Lambda$ if there exist constants $\alpha, \rho \in \mathbb{R}$, with $\rho > \alpha$, and $\kappa \in (0,1]$ such that 
		\begin{subequations}
			\begin{align}\label{sys2}
				&\mathcal B(x) \leq \alpha,\quad\quad\quad\quad\quad\quad\quad\quad\quad\quad\!\forall x \in X_{0},\\\label{sys3}
				&\mathcal B(x) \geq \rho, \quad\quad\quad\quad\quad\quad\quad\quad\quad\quad\forall x \in X_{u},\\\label{BCeq}
				& \mathcal B(f(x)) \leq\kappa\mathcal B(x), \quad\quad\quad\quad\quad\quad\quad\!\!\forall x \in X.
			\end{align} 
		\end{subequations}
	\end{definition}
	
	The following theorem, borrowed from~\citep{prajna2004safety}, ensures that the system trajectories do not enter the unsafe region.
	
	\begin{theorem}\label{Kushner}
		Given a dt-NS $\Lambda = (X, f)$, suppose that $\mathcal{B}$ is a BC for $\Lambda$. Then, for any $x_0 \in X_0$ and $k \in \mathbb{N}_0$, the state trajectory $x_{x_0}(k)$ does not enter the unsafe region $X_u$, i.e., $x_{x_0}(k) \notin X_u$, over an infinite time horizon.
	\end{theorem}
	
	Ensuring the safety of the dt-NS in \eqref{Eq_1a}, as outlined in Theorem~\ref{Kushner}, relies on precise knowledge of $f$ to verify condition \eqref{BCeq}. However, since this information is unavailable in our context, in the next section we briefly present an existing solution based on the scenario approach~\citep{nejati2023formal}, which facilitates the design of BC using finite data derived from observed system trajectories.
	
	\section{Scenario Approach for BC Design}\label{data-driven}
	
	Consider the structure of the BC as  $\mathcal{B}(q, x) = \sum_{j=1}^{z} q^j {l}^j(x)$, where ${l}^j$ represent user-defined (potentially nonlinear) basis functions, and $q = [q^1; \dots; q^{z}] \in \mathbb{R}^{z}$ are the unknown coefficients. Given that $\mathcal{B}(q, x)$ is our selected choice, we assume its basis functions are chosen such that $\mathcal{B}(q, x)$ is continuously differentiable.
	
	To ensure that conditions~\eqref{sys2}-\eqref{BCeq} are fulfilled, the safety problem can be formulated as a robust optimization program (ROP)~\citep{nejati2023formal}:
	\begin{subequations}\label{Eq_ROP}
		\begin{align}
			\min_{[\mathsf d;\eta]}\!\!\!&\quad\quad \eta &  \notag \\
			\mathrm{{\bf s.t.}}&\quad\quad\mathcal  B(q, x) - \alpha \leq \eta,\!\quad\quad\quad\quad\quad\quad~\!\forall x \in X_{0}, \label{Eq_ROP2}\\
			&\!\quad -\mathcal B(q, x) + \rho \leq \eta, \quad\quad\quad\quad\quad~~~\forall x \in X_{u}, \label{Eq_ROP3} \\
            &\quad\quad  \alpha-\rho-\zeta \leq 0, \label{Eq_ROP_LevelSet} \\
			&\quad~~~\mathcal B(q, f(x)) - \kappa\mathcal B(q, x)  \leq \eta,\quad\quad\!\!\! ~\!\forall x\in X , \label{Eq_ROP4}\\\notag
			&\quad~~~~\!\! {\mathsf d}=[\alpha; \rho; q^1;\dots; q^{z}],~\alpha, \rho, \eta, q^{j} \!\in\! \R, ~\kappa \!\in\!(0,1], \zeta\!<\!0.
		\end{align}
	\end{subequations}
	With some $\zeta<0$ in \eqref{Eq_ROP_LevelSet}, it is guaranteed that $\alpha < \rho$. The optimal value of the ROP is denoted as $\eta_{R}^*$. If $\eta_{R}^* \leq 0$, solving the ROP confirms that conditions~\eqref{sys2}-\eqref{BCeq} are satisfied. Note that the ROP in~\eqref{Eq_ROP} is convex with respect to the decision variables, owing to the structure of $\mathcal{B}(q, x) = \sum_{j=1}^{z} q^j {l}^j(x)$, with only mild bilinearity appearing between $\kappa$ and $q^z$ in \eqref{Eq_ROP4}. To resolve this, given that $\kappa$ lies between $0$ and $1$, it can be pre-selected when solving the ROP.
	
	Given that knowledge of $f(x)$ is still required in \eqref{Eq_ROP4}, and recognizing that the evolution of the dt-NS unfolds {recursively}, data $\{\hat{x}^s\}_{s=1}^S$ with $S \in \mathbb{N}$, can be collected by treating the first data point as the state and the second as the unknown map $f(x)$. This allows the ROP to be reformulated as a scenario optimization program (SOP):
	\begin{subequations}\label{Eq_SOP_ALLDATA}
		\begin{align}
			\!\!\min_{[\mathsf d;\eta]}\!\!\!&\quad \quad\eta &  \notag \\
			\mathrm{{\bf s.t.}}&\quad\quad\mathcal B(q, \hat{x}^s) \!-\! \alpha \leq \eta,\quad\quad\quad\quad\quad\quad\quad\quad~\!\!\!\!\!\!\!\!\forall \hat{x}^s \!\in\! X_{0}, \!\!\label{Eq_SOP_ALLDATA2}\\
			&\quad\! -\mathcal B(q, \hat{x}^s) \!+\! \rho \leq \eta, \quad\quad\quad\quad\quad\quad\quad~\!\forall \hat{x}^s \quad\!\!\!\!\!\!\!\in\! X_{u}, \!\!\!\label{Eq_SOP_ALLDATA3} \\
            &\quad\quad  \alpha-\rho-\zeta \leq 0, \label{Eq_SOP_LevelSet} \\
			&\quad\quad\mathcal B(q, f(\hat{x}^s)) - \kappa\mathcal B(q, \hat{x}^s) \leq  \eta,\quad\quad~~\!\!\!\forall \hat{x}^s\quad\!\!\!\!\!\!\!\in\! X,\label{Eq_SOP_ALLDATA4}\\
			&\quad\quad \forall s\in\{1,\dots,S\}, ~{\mathsf d}=[\alpha; \rho; q^1;\dots; q^{z}], \notag\\& \qquad\alpha, \rho, \eta, q^{j} \!\in\! \R, ~\kappa \!\in\!(0,1], \zeta < 0. \notag
		\end{align}
	\end{subequations}
    The optimal value of the SOP is denoted by $\eta_S^*$.
	
	In existing studies within the relevant literature \citep[\emph{e.g.,}~][]{nejati2023formal,aminzadeh2024compositional}, the SOP can be solved using finite data, with the results transferred back to the ROP while providing {out-of-sample performance guarantees}. However, the number of data required for solving the SOP is extremely high due to the inherent \emph{exponential} sample complexity of the scenario approach with respect to the size of the state space. Motivated by this challenge, we introduce in the next section a \emph{physics-informed} scenario approach that leverages the physical principles of the underlying dynamics to substantially reduce the dataset size required to solve the SOP.
	
	\section{Physics-Informed Scenario Approach}\label{physics-informed data-driven}
	We consider discrete-time nonlinear systems based on physical laws as 
	\begin{equation}\label{Eq_ph}
		\Lambda^{phy}\!:x(k+1) = {f^{phy}(x(k))}, \quad k\in\mathbb N_0.
	\end{equation}
	Such models can be obtained based on {fundamental physical laws}. For instance, \emph{electrical} circuits can be modeled using Kirchhoff's voltage and current laws, while \emph{mechanical} systems can be described through Newton's second law of motion. We assume access to data $\{\hat{x}^s\}_{s=1}^{S}$ collected from trajectories of the unknown model $f$ in \eqref{Eq_1a}, which captures the true behavior of the system, as well as a physics-informed model in \eqref{Eq_ph} based on fundamental physical laws, providing valuable information about the system. 
	
	Our proposed physics-informed framework guides the SOP in \eqref{Eq_SOP_ALLDATA} to utilize only those samples where the {physics-informed} one-step dynamics evolution (\emph{i.e.,} $f^{phy}(\hat{x}^{s})$) is sufficiently close (within a specified threshold) to the one-step evolution of the unknown model (\emph{i.e.,} $f(\hat{x}^{s})$). In particular, we construct a new dataset $\hat{X}^{p}$ from the available dataset $\{\hat{x}^s\}_{s=1}^{S}$ within $X$, by selecting only those samples that satisfy
	\begin{align}\label{closedata}	
		\hat{X}^{p} = \Big\{\hat{x}^{p} \in \{\hat{x}^s\}_{s=1}^{S} \mid \Vert f^{phy}(\hat{x}^{s}) - f(\hat{x}^{s}) \Vert \leq \delta\Big\},
	\end{align}
	where $\delta \in \mathbb{R}_{>0}$ is a sufficiently small threshold that captures the closeness between the behaviors of the two systems. This threshold can potentially guide the scenario approach to retain only those samples where the behaviors of the two systems are similar, while discarding the remaining samples as redundant. We denote the cardinality of the set $\hat{X}^{p}$ as $P \in \mathbb{N}$.

We note that the physics-based model is assumed to be informative but not fully accurate; thus, safety verification relies primarily on data, while the model is used only to guide data reduction.
	Now, instead of solving the SOP in \eqref{Eq_SOP_ALLDATA}, we propose the following {physics-informed} SOP, denoted by SOP$_{phy}$:
	\begin{subequations}\label{Eq_SOP}
		\begin{align}
			\!\!\min_{[\mathsf d;\eta]}\!\!\!&\quad\quad \eta &  \notag \\
			\mathrm{{\bf s.t.}}&\!\quad~~~\mathcal B(q, \hat{x}^p) \!-\! \alpha \!\leq\!  \eta,\quad\quad\quad\quad\quad\quad\quad~\!\!\!\forall \hat{x}^p \!~\!\!\in X_{0}, \!\!\label{Eq_SOP2}\\
			&\quad\!\! -\mathcal B(q, \hat{x}^p) \!+\! \rho \leq \eta, \quad\quad\quad\quad\quad\quad\quad~\!\!\!\!\!\forall \hat{x}^p \quad\!\!\!\!\!\!\!\in\! X_{u}, \!\!\!\label{Eq_SOP3} \\
            &\quad\quad  \alpha-\rho-\zeta \leq 0, \label{Eq_SOP_Phy_LevelSet} \\
			&\!\quad~~~\mathcal B(q, f(\hat{x}^p)) - \kappa\mathcal B(q, \hat{x}^p) \leq  \eta,\quad\!~~\forall \hat{x}^p\!\in X,\label{Eq_SOP4} \\
			&\quad~~~\! \forall p\in\{1,\dots,P\}, ~{\mathsf d}=[\alpha; \rho; q^1;\dots; q^{z}],\notag\\ &\qquad \alpha, \rho, \eta, q^{j} \!\in\! \R, ~\kappa \!\in\!(0,1], \zeta < 0. \notag
		\end{align}
	\end{subequations}
	We denote the optimal value of SOP$_{phy}$ by $\eta^*_{phy}$. 
\begin{remark}\label{remark_limitation}
A key limitation of the current approach is that the physics-informed model $f^{phy}$ and the threshold $\delta$ are used solely for sample selection in \eqref{closedata} and do not explicitly appear in the SOP constraints in \eqref{Eq_SOP}. As a result, while this strategy successfully mitigates data requirements in the case studies of Section \ref{case study}, it does not explicitly capture the effect of $f^{phy}$ in the optimization problem. We acknowledge this limitation and address it in a subsequent study, inspired by this work, by directly integrating $f^{phy}$ and $\delta$ into the scenario approach conditions to formally capture the model accuracy within the safety guarantees.
\end{remark}

In the following subsections, we present our main results for constructing BC by solving the proposed {physics-informed SOP} in \eqref{Eq_SOP} using the reduced amount of data that satisfies~\eqref{closedata}.

	\subsection{Correctness Guarantee with Less Data}\label{safety-ensurance}
	To establish our results, we first assume that $f(x)$ is Lipschitz continuous with respect to $x$, a standard assumption in the scenario approach. Given that $\mathcal{B}(q, x)$ is continuously differentiable and our analysis is conducted on the bounded domain $X$, it follows that $\mathcal{B}(q^*\!, f(x)) - \kappa^* \mathcal{B}(q^*\!, x)$ is also Lipschitz continuous with respect to $x$, with a Lipschitz constant $\mathscr{L}^2$. Similarly, using the same reasoning, it can be shown that $\mathcal{B}(q^*\!, x)$ is Lipschitz continuous with respect to $x$, with a Lipschitz constant $\mathscr{L}^1$.

 Samples $\{\hat{x}^s\}_{s=1}^{S}$ are first collected using a uniform grid size over the state space. Our physics-informed data-driven method is then applied to $\{\hat{x}^s\}_{s=1}^{S}$  using condition \eqref{closedata} to construct the physics-informed set $ \{\hat{x}^p\}_{p=1}^P$, and accordingly, guide the SOP in \eqref{Eq_SOP_ALLDATA}. This procedure ensures that the focus is on samples which, after one step of evolution, stay within a specified threshold of the output computed by the physics-based dynamics. By doing so, our approach  reduces the number of samples needed to solve the SOP, allowing us to address  SOP$_{phy}$ \eqref{Eq_SOP} instead of the more computationally demanding SOP \eqref{Eq_SOP_ALLDATA}.
	
	To achieve this, we consider physics-informed samples \(\hat{X}^{p}  = \{\hat{x}^p\}_{p=1}^P\), each associated with a ball of radius \(\epsilon^{p}\) around the sample $\hat{x}^p$, denoted as $\mathsf{X}^p$, such that \(X \subseteq \cup_{p=1}^{P} \mathsf{X}^p\), and
	\begin{align}\label{New1}
		\left \lVert x - \hat{x}^p\right \rVert \leq \epsilon^{p} , \quad\forall x \in \mathsf{X}^p, ~\forall p\in\{1,\dots,P\}.
	\end{align}
	
	The following theorem, inspired by~\cite{nejati2023data}, offers our physics-informed scenario approach with a correctness guarantee.
	
	\begin{theorem} \label{confidence}
		Given a dt-NS \(\Lambda = (X, f)\), let us solve the SOP$_{phy}$ in~\eqref{Eq_SOP} using the physics-informed dataset $\hat{X}^{p}$, which is constructed according to~\eqref{closedata} and covers the state space $X$ with associated $\epsilon^p$ satisfying~\eqref{New1}. If
		\begin{align} \label{main condition}
			&\mathscr L\epsilon^{\max}+\eta^{*}_{phy}\leq 0,
		\end{align}
		with $\mathscr L =\max\{\mathscr L^1, \mathscr L^2\}$, and $ \epsilon^{\max} = \max\{\epsilon^1,\dots, \epsilon^P\}$, then $\mathcal B$ obtained by solving the SOP$_{phy}$ in \eqref{Eq_SOP} is a BC for $\Lambda$ with a correctness guarantee.  
	\end{theorem}
	
	\begin{proof}
		We first demonstrate that, under condition \eqref{main condition}, $\mathcal B$ constructed by solving SOP$_{phy}$ in \eqref{Eq_SOP} satisfies \eqref{Eq_ROP4} across the entire state space $X$. Note that according to~\eqref{New1}, for any $x\in X$, there exists $\hat x^p \in \mathsf{X}^p$ such that $x$ and $\hat x^p$ are $\epsilon^p$-close. Given that $\mathcal{B}(q^*\!, f(x)) - \kappa^* \mathcal{B}(q^*\!, x)$ is Lipschitz continuous with respect to $x$ with Lipschitz constant  $\mathscr L^2$, adding and subtracting $\mathcal B(q^{*},f(\hat x^{p})) - \kappa^{*}\mathcal B(q^{*}, \hat x^{p})$ yields
		\begin{align*}
			&\mathcal B(q^{*}, f(x)) - \kappa^{*} \mathcal B(q^{*}, x) =\\
            &\mathcal B(q^{*}, f(x))- \kappa^{*}\mathcal B(q^{*}, x) - \bigr(\mathcal B(q^{*},f(\hat x^{p})) - \kappa^{*}\mathcal B(q^{*}, \hat x^{p})\bigr) \notag \\
			&+ \underbrace{\bigr(\mathcal B(q^{*}, f(\hat x^{p})) - \kappa^{*}\mathcal B(q^{*}, \hat x^p)\bigr)}_{\leq \eta^*_{phy}}\leq \mathscr L^2\Vert x-\hat x^p\Vert+\eta^*_{phy}\\&\leq \mathscr L\Vert x-\hat x^p\Vert+\eta^*_{phy} \overset{\eqref{New1}}{\leq} \underbrace{\mathscr L\epsilon^{\max}+\eta^*_{phy}}_{=\eta_{R}^*}\overset{\eqref{main condition}}{\leq}0.
		\end{align*}
		Thus, by defining $\eta_{R}^* = \mathscr L\epsilon^{\max}+\eta^*_{phy}$, the constructed $\mathcal{B}$ obtained by solving the SOP$_{phy}$ in \eqref{Eq_SOP} satisfies \eqref{Eq_ROP4} over the entire state space $X$. By employing similar reasoning and adding and subtracting $\mathcal B(q^{*}, \hat x^{p})$, one can demonstrate that under condition \eqref{main condition}, the constructed $\mathcal{B}$ resulting from solving the SOP$_{phy}$ in \eqref{Eq_SOP} fulfills conditions ~\eqref{Eq_ROP2} and \eqref{Eq_ROP3} for the sets $X_{0}$ and  $X_{u}$, respectively: 
		\begin{align*}
			&\mathcal B(q^{*}, x) - \alpha^* =  \mathcal B(q^{*}, x) - \mathcal B(q^{*}, \hat x^{p})+ \underbrace{\mathcal B(q^{*}, \hat x^p) - \alpha^* }_{\leq \eta^*_{phy}}\\&
			\leq \mathscr L^1\Vert x-\hat x^p\Vert+\eta^*_{phy}\leq \mathscr L\Vert x-\hat x^p\Vert+\eta^*_{phy} \\&\overset{\eqref{New1}}{\leq} \underbrace{\mathscr L\epsilon^{\max}+\eta^*_{phy}}_{=\eta_{R}^*}\overset{\eqref{main condition}}{\leq}0, ~~\text{and}\\
			&-\mathcal B(q^{*}, x)+ \rho^* =  \mathcal B(q^{*}, \hat x^p)- \mathcal B(q^{*}, x)  \underbrace{- \mathcal B(q^{*}, \hat x^{p})
				+ \rho^* }_{\leq \eta^*_{phy}}\\& \leq \mathscr L^1\Vert x-\hat x^p\Vert+\eta^*_{phy}\leq \mathscr L\Vert x-\hat x^p\Vert+\eta^*_{phy} \\&\overset{\eqref{New1}}{\leq} \underbrace{\mathscr L\epsilon^{\max}+\eta^*_{phy}}_{=\eta_{R}^*}\overset{\eqref{main condition}}{\leq}0. 
		\end{align*}
Furthermore, with $\zeta^*<0$, the condition \eqref{Eq_SOP_Phy_LevelSet} ensures $\rho^*>\alpha^*$. Therefore, $\mathcal B$ obtained by solving the SOP$_{phy}$ in \eqref{Eq_SOP} is a BC for $\Lambda$, thereby completing the proof.
	\end{proof}
	
	To verify condition \eqref{main condition} in Theorem \ref{confidence}, it is necessary to compute $\mathscr{L}^1$ and  $\mathscr{L}^2$. To accomplish this, existing methods in the literature can be employed to compute the Lipschitz constants $\mathscr{L}^1$ and  $\mathscr{L}^2$ from the collected data (\emph{e.g.,} \citealp[Algorithm 1]{knuth2021planning} or \citealp[Algorithm 2]{nejati2023formal}), which leverage the fundamental result of~\citep{wood1996estimation}.

	\section{Discussion} \label{discussion}
	
	In scenario-based approaches, while a larger dataset typically leads to a less negative \(\eta^*_S\) (indicating a more conservative outcome), our approach avoids solving the optimization problem for the entire dataset. Instead, the physics-based model strategically selects specific samples, enabling the problem to be solved in a more efficient manner. Nevertheless, we note that the sample complexity in our work still scales exponentially with the state dimension. A key distinction between our work and existing sampling approaches, such as \citep{aminzadeh2024compositional}, is that increasing $\epsilon^p$ in \eqref{New1} using uniform grid sampling reduces the number of samples but may fail to satisfy the condition in \eqref{main condition} due to the {uniformity of the samples}. In contrast, our physics-informed approach does not require uniform sampling; instead, it guides the sampling process through condition \eqref{closedata}, allowing some samples to be closely clustered while others are more widely spaced. This selective sampling strategy involves ignoring certain points within the state space, focusing only on those that have richness for further analysis, guided by the physical principles of the underlying dynamics. Figure~\ref{fig:gridjump} in the case study section clearly demonstrates this concept.

\section{Case Study} \label{case study}
\captionsetup{justification=raggedright, singlelinecheck=false}

\begin{table*}[t!]
    \centering
    \caption{Comparison of the required data for the proposed physics-informed scenario approach and the traditional approach~\citep{aminzadeh2024compositional}, where “condition” refers to~\eqref{main condition}.}
    \resizebox{\textwidth}{!}{%
    \begin{tabular}{|l|c|c|c|c|c|c|c|c|c|c|}
    \hline
    \textbf{Case Study} & \textbf{Approach} & \textbf{samples} & \textbf{$\delta$} & \textbf{$\epsilon^{\max}$} & \textbf{$\varphi$} & \textbf{$\mathscr{L}$} & \textbf{$\eta$} & \textbf{\% change ($\eta$)} & \textbf{condition} & \textbf{\% change (condition)} \\ \hline
    
    % ========== Supply-Demand ==========
    \multirow{2}{*}{\textbf{Supply-Demand}} 
    &\cellcolor{red!20} \textbf{Traditional} 
    &\cellcolor{red!20} 220{,}000 
    &\cellcolor{red!20} - 
    &\cellcolor{red!20} \(5 \times 10^{-6}\) 
    &\cellcolor{red!20} - 
    &\cellcolor{red!20} 0.1308 
    &\cellcolor{red!20} \(-2.3310 \times 10^{-3}\) 
    &\cellcolor{red!20} - 
    &\cellcolor{red!20} \(-2.3303 \times 10^{-3}\) 
    &\cellcolor{red!20} - \\ 
    
    & \cellcolor{green!20} \textbf{Physics-informed} 
    &\cellcolor{green!20} 109{,}971 
    &\cellcolor{green!20} 0.005 
    &\cellcolor{green!20} \(8 \times 10^{-5}\) 
    &\cellcolor{green!20} - 
    &\cellcolor{green!20} 0.1188 
    &\cellcolor{green!20} \(-5.2547 \times 10^{-3}\) 
    &\cellcolor{green!20} \(125\% \text{ more negative}\) 
    &\cellcolor{green!20} \(-5.2452 \times 10^{-3}\) 
    &\cellcolor{green!20} \(125\% \text{ more negative}\) \\ \hline
    
    % ========== Logistic Growth ==========
    \multirow{2}{*}{\textbf{Logistic Growth}} 
    &\cellcolor{red!20} \textbf{Traditional} 
    &\cellcolor{red!20} 90{,}000 
    &\cellcolor{red!20} - 
    &\cellcolor{red!20} \(5 \times 10^{-6}\) 
    &\cellcolor{red!20} - 
    &\cellcolor{red!20} 0.2229
    &\cellcolor{red!20} \(-6.4247 \times 10^{-4}\) 
    &\cellcolor{red!20} - 
    &\cellcolor{red!20} \(-6.4136 \times 10^{-4}\) 
    &\cellcolor{red!20} - \\ 
    
    & \cellcolor{green!20} \textbf{Physics-informed} 
    &\cellcolor{green!20} 45{,}138 
    &\cellcolor{green!20} 0.005 
    &\cellcolor{green!20} \(7.5 \times 10^{-5}\) 
    &\cellcolor{green!20} - 
    &\cellcolor{green!20} 0.3415
    &\cellcolor{green!20} \(-2.0617 \times 10^{-3}\) 
    &\cellcolor{green!20} \(221\% \text{ more negative}\) 
    &\cellcolor{green!20} \(-2.0361 \times 10^{-3}\) 
    &\cellcolor{green!20} \(218\% \text{ more negative}\)  \\ \hline

    % ========== Jet Engine ==========
    \multirow{2}{*}{\textbf{Jet Engine}} 
    &\cellcolor{red!20} \textbf{Traditional} 
    &\cellcolor{red!20} 810{,}000 
    &\cellcolor{red!20} - 
    &\cellcolor{red!20} \(7.07 \times 10^{-4}\)
    &\cellcolor{red!20} - 
    &\cellcolor{red!20} 0.2733 
    &\cellcolor{red!20} \(-0.0229\)
    &\cellcolor{red!20} - 
    &\cellcolor{red!20} \(-0.0227\)
    &\cellcolor{red!20} - \\ 
    
    & \cellcolor{green!20} \textbf{Physics-informed} 
    &\cellcolor{green!20} 407{,}911 
    &\cellcolor{green!20} 0.0008 
    &\cellcolor{green!20} \(0.0056\)
    &\cellcolor{green!20} - 
    &\cellcolor{green!20} 0.2737 
    &\cellcolor{green!20} \(-0.0253\)
    &\cellcolor{green!20} \(11\% \text{ more negative}\)
    &\cellcolor{green!20} \(-0.0238\)
    &\cellcolor{green!20} \(5\% \text{ more negative}\) \\ 
    \hline
    
    \end{tabular}%
    }
    \label{tab:case-study}
\end{table*}

\subsection{Supply-Demand} 

Consider the following {physics-informed} supply-demand model
\begin{align*}
	\Lambda^{phy}\!:x(k+1)=x(k)+0.1(5-2x(k)), \quad k\in\mathbb N_0,
\end{align*}
where $x(k)$ represents the price at time step \(k\), and \(x(k+1)\) is the updated price determined by supply and demand dynamics. The ``adjustment factor'' \(0.1\) controls the rate at which the price evolves over time. The term \(5 - 2x(k)\) represents the ``demand-supply balance'', reflecting the desired price level based on demand. 
We are also provided with the data $\{\hat{x}^s\}_{s=1}^S$, over which we aim to perform safety analysis. The regions of interest are given as \( X \in [0.5, 2.7] \), \( X_0 \in [0.5, 0.6] \), and \( X_u \in [2.6, 2.7] \). To design the BC, we consider its structure as \( \mathcal{B}(q,x) = q_1 x^2 + q_2 x + q_3 \). We are given a sample size of $S= 220,000$. We choose \( \delta = 0.005 \) to construct the {physics-informed} data set as specified in \eqref{closedata}, resulting in \( P = 109,971 \). By setting \( \kappa^*= 0.83 \), we solve the $\text{SOP}_{{phy}}$ and compute the BC coefficients along with the other decision variables as
\begin{align*} 
&\mathcal{B}(q,x) = 0.2 x^2 + 0.1054 x -1.6155,\\& \alpha^* \!=\! 0.0001, ~\rho^* \!=\! 0.0054, ~\eta_{phy}^* \!\!=\! -0.0053.
\end{align*}
We compute \( \epsilon^{\max} = 8 \times 10^{-5} \) based on the physics-informed sampling and are given \( \mathscr{L} = 0.1188 \). Given that \(\mathscr{L} \epsilon^{\max} + \eta_{{phy}}^* = -0.0052 \leq 0 \), as required by Theorem \ref{confidence}, we conclude that the {physics-informed} data-driven BC is valid for the unknown system \( \Lambda \) across the entire state space, guaranteeing its correctness. The {physics-informed} data-driven BC is depicted in Fig.~\ref{fig:linear} (a).

To demonstrate the effectiveness of our proposed method in mitigating data requirements, we compare it with traditional scenario-based approach from the literature that retain all samples for analysis \citep[\emph{e.g.,}][]{aminzadeh2024compositional}. As shown in Table~\ref{tab:case-study}, while the traditional scenario-based approach requires $220,000$ data points to solve this problem, our physics-informed approach achieves the solution with only $109,971$ data points. Notably, despite requiring fewer data points, the main condition \eqref{main condition} in our approach is even $125\%$ \emph{more negative} than that of the traditional approach by \cite{aminzadeh2024compositional}, demonstrating the effectiveness of our method. The physics-informed sampling strategy in our approach according to~\eqref{closedata} is partially depicted in Fig.~\ref{fig:gridjump} (a), focusing on the region where the {maximum jump} in sampling occurs.

\begin{figure}[t!]
\begin{center}	\includegraphics[width=5.6cm]{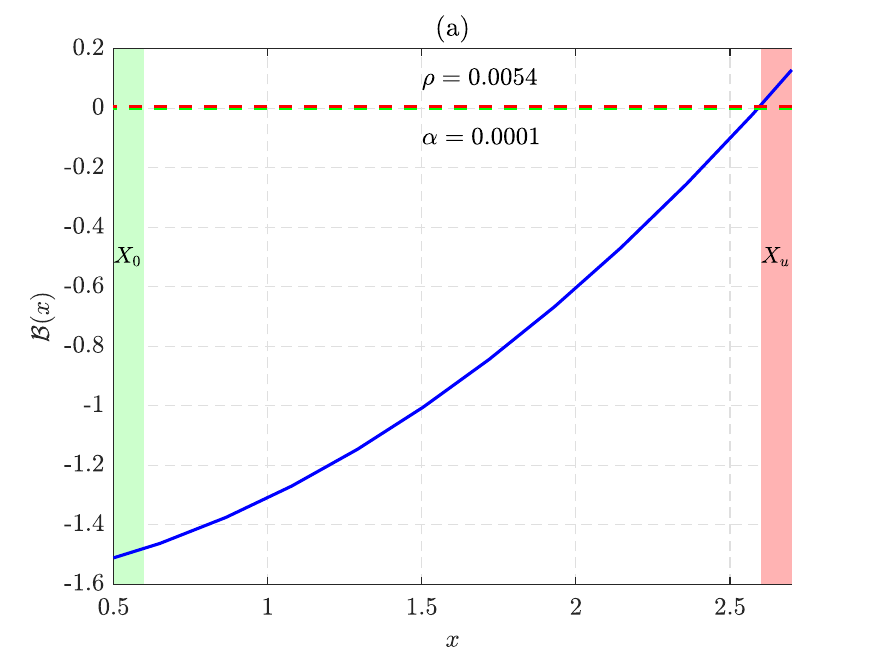}\\\includegraphics[width=5.6cm]{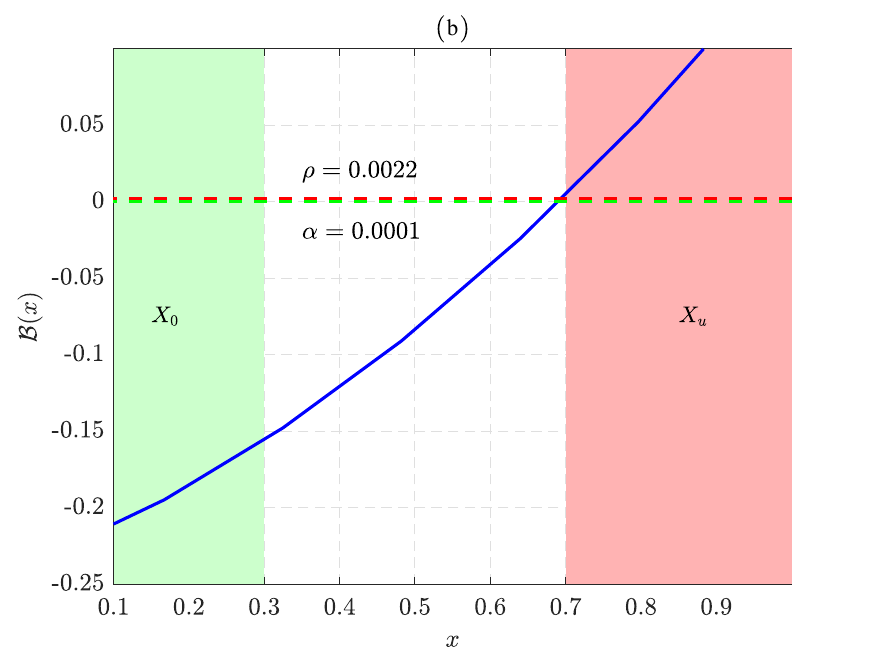}\\\includegraphics[width=5.6cm]{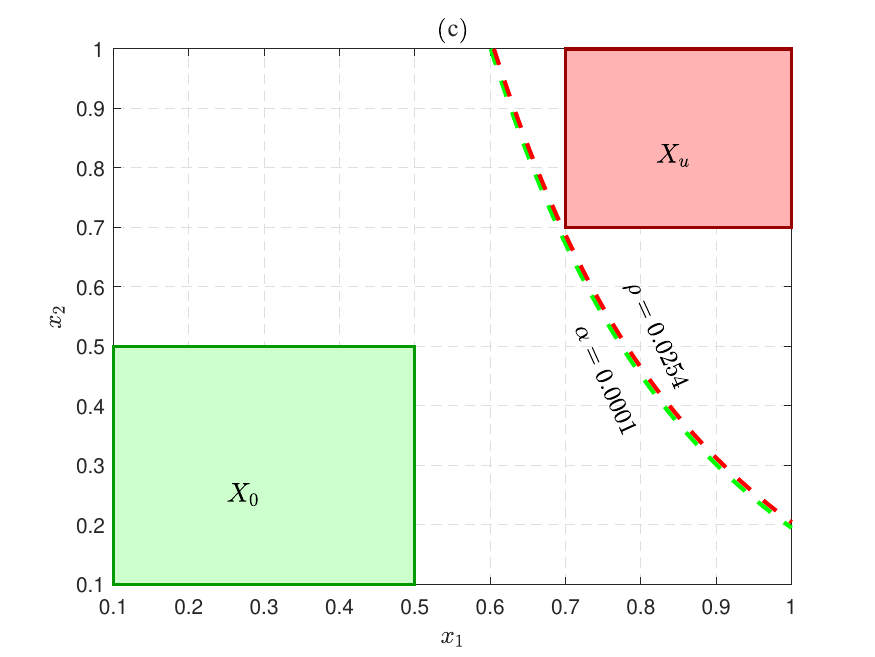}
\caption{{Physics-informed} BC for (a) supply–demand, (b) logistic-growth, and (c) jet engine case studies. The quadratic BC segment is shown in blue, satisfying \eqref{sys2} and \eqref{sys3}; green and red dashed lines denote the initial and unsafe level sets. The green and red regions denote the initial and unsafe sets.}
  \label{fig:linear}
  \end{center}
\end{figure}
\begin{figure}[t!]
	\begin{center}
		\includegraphics[width=7.5cm]{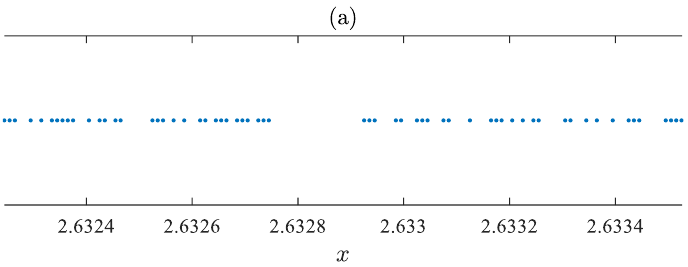}\vspace{0.1cm}
		\includegraphics[width=7.5cm]{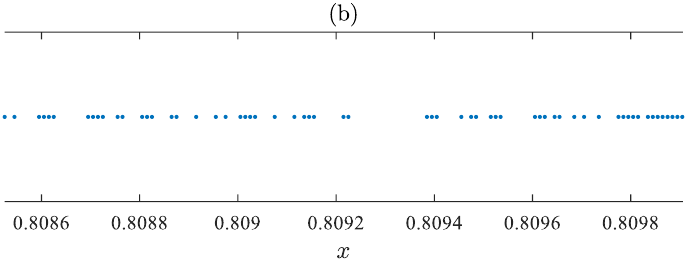}
		\caption{{Physics-informed} sampling strategy from~\eqref{closedata} for (a) supply–demand and (b) logistic-growth examples, highlighting the region of maximal sampling jump.}
		\label{fig:gridjump}
	\end{center}
\end{figure}

\subsection{Logistic Growth}
Consider the following physics-informed \emph{nonlinear} logistic growth model, adapted from \citep{anokye2024logistic},
\begin{align*} 
	\Lambda^{phy}\!:x(k+1) \!=\! x(k)\!+\!0.5x(k)(1\!-\!x(k))\!-\!0.2x(k), ~ k\!\in\!\mathbb N_0,
\end{align*}
where \( x(k) \) denotes the population at time step \( k \), the term \( 0.5x(k)(1 - x(k)) \) models logistic growth with a carrying capacity, and the term \( -0.2x(k) \) represents a linear damping (or decay) effect, which reduces the population growth rate over time. We are also provided with the data $\{\hat{x}^s\}_{s=1}^S$ to perform safety analysis. The regions of interest are as follows: \( X \in [0.1, 1] \), \( X_0 \in [0.1, 0.3] \), and \( X_u \in [0.7, 1] \). We are given a sample size of $S= 90,000$. We select \( \delta = 0.005 \) to construct the {physics-informed} data set in \eqref{closedata}, resulting in $P = 45,138$. By setting \( \kappa^* = 0.83 \), we solve the $\text{SOP}_{{phy}}$ and compute the BC coefficients along with the other decision variables as 
\begin{align*} 
&\mathcal{B}(q,x) = 0.2 x^2 + 0.2 x -0.2338,\\& \alpha^* \!=\! 0.0001, \rho^* \!=\! 0.0022, \eta_{phy}^* \!\!=\! -0.0021.
\end{align*}
We compute \( \epsilon^{\max} = 7.5 \times 10^{-5} \) and  are given \( \mathscr{L} = 0.3415 \). Since \( \mathscr{L} \epsilon^{\max} + \eta_{{phy}}^* = -0.0020 \leq 0 \), we conclude that the {physics-informed} data-driven BC is valid for the unknown system \( \Lambda \) across the entire state space with a guarantee of correctness. The {physics-informed} BC derived from the data is depicted in Fig.~\ref{fig:linear} (b). We showcase the data efficiency of our approach compared to the traditional scenario-based approach from the literature \citep[\emph{e.g.,}][]{aminzadeh2024compositional} in Table~\ref{tab:case-study}. Fig.~\ref{fig:gridjump} (b) illustrates our selective physics-informed sampling strategy, highlighting the area where the {maximum jump} in sampling happens.

\subsection{Jet Engine}
Consider the following physics-informed \emph{nonlinear} Moore-Greitzer jet engine model \citep{krstic1995lean}:
\begin{align*} 
\Lambda^{phy}\!\!:
\begin{cases}
\!{x}_1(k+1) = x_1(k)\!+\!\tau(\!-x_2(k)\! -\!\! 1.5 x_1^2(k)\! -\! 0.5 x_1^3(k)),\\ 
\!{x}_2(k+1) = x_2(k)+\tau x_1(k),
\end{cases}
\end{align*}
where \(x_1 = \Phi-1\) and \(x_2 = \Psi - \Theta-2\). The quantities \(\Phi\), \(\Psi\), and \(\Theta\) denote the mass flow, pressure rise, and a constant parameter, respectively, and \(\tau\) is the  sampling time. The regions of interest are $X = [0.1,1]^2, X_0 = [0.1, 0.5]^2, X_u = [0.7, 1]^2.$ We choose the BC structure 
$\mathcal{B}(q,x) = q_1 x_1^2 + q_2 x_2^2 + q_3 x_1 + q_4 x_2 + q_5 x_1x_2 + q_6$.
Given $S = 810{,}000$ samples and selecting $\delta = 0.0008$ in~\eqref{closedata}, we obtain $P = 407{,}911$ physics-informed samples.  
With $\kappa^* = 0.95$, solving $\mathrm{SOP}_{\mathrm{phy}}$ yields
\begin{align*} &\mathcal{B}(q,x)\! =\! 0.5 x_1^2\! + \!0.5 x_2^2 \!+\!0.5 x_1 \!-\!2.4726x_2\!+\!5x_1x_2\!-\!1.5117,\\& \alpha^* \!=\! 0.0001, \rho^* \!=\! 0.0254, \eta_{phy}^* \!\!=\! -0.0253. \end{align*}
We compute $\epsilon^{\max} = 0.0056$ and are given $\mathscr{L} = 0.2737$. Since
$\mathscr{L}\epsilon^{\max} + \eta_{\mathrm{phy}}^* = -0.0238 \le 0$,
Theorem~\ref{confidence} guarantees that the physics-informed BC is valid for the unknown system~$\Lambda$ over the entire state space. 
The resulting BC appears in Fig.~\ref{fig:linear} (c).  
Table~\ref{tab:case-study} highlights the data efficiency of our approach compared to traditional scenario-based methods \citep[\emph{e.g.,}][]{aminzadeh2024compositional}.

\section{Conclusion}\label{conclusion}
In this paper, we developed a {physics-informed} scenario approach for verifying the safety of nonlinear systems via BCs. Traditional scenario methods, though effective, often require large datasets due to their {exponential} sample complexity. Our physics-informed approach addressed this by selecting data samples whose outputs closely match those of a physics-based model, reducing redundancy and  shrinking the required dataset. Three case studies  demonstrated the ability of our approach to achieve safety guarantees with fewer samples. Extending this approach to controller synthesis is a direction for future work.

\bibliography{ifacconf}

\begin{thebibliography}{27}
\providecommand{\natexlab}[1]{#1}
\providecommand{\url}[1]{\texttt{#1}}
\providecommand{\urlprefix}{URL }
\expandafter\ifx\csname urlstyle\endcsname\relax
  \providecommand{\doi}[1]{doi:\discretionary{}{}{}#1}\else
  \providecommand{\doi}{doi:\discretionary{}{}{}\begingroup
  \urlstyle{rm}\Url}\fi

\bibitem[{Ames et~al.(2019)Ames, Coogan, Egerstedt, Notomista, Sreenath, and
  Tabuada}]{ames2019control}
Ames, A.D., Coogan, S., Egerstedt, M., Notomista, G., Sreenath, K., and
  Tabuada, P. (2019).
\newblock Control barrier functions: Theory and applications.
\newblock In \emph{18th European control conference (ECC)}, 3420--3431.

\bibitem[{Aminzadeh et~al.(2024)Aminzadeh, Swikir, Haddadin, and
  Lavaei}]{aminzadeh2024compositional}
Aminzadeh, A., Swikir, A., Haddadin, S., and Lavaei, A. (2024).
\newblock Compositional safety verification of infinite networks: A data-driven
  approach.
\newblock In \emph{European Control Conference (ECC)}, 545--551.

\bibitem[{Anokye(2024)}]{anokye2024logistic}
Anokye, M. (2024).
\newblock A logistic growth model with discrete-time delay and a restriction on
  harvesting.
\newblock \emph{Journal of Mathematics}, 2024(1).

\bibitem[{Banse et~al.(2024)Banse, Romao, Abate, and Jungers}]{banse2024data}
Banse, A., Romao, L., Abate, A., and Jungers, R. (2024).
\newblock Data-driven memory-dependent abstractions of dynamical systems via a
  {C}antor-{K}antorovich metric.
\newblock \emph{arXiv:2405.08353}.

\bibitem[{Calafiore and Campi(2006)}]{calafiore2006scenario}
Calafiore, G.C. and Campi, M.C. (2006).
\newblock The scenario approach to robust control design.
\newblock \emph{IEEE Transactions on automatic control}, 51(5), 742--753.

\bibitem[{Campi et~al.(2009)Campi, Garatti, and Prandini}]{campi2009scenario}
Campi, M.C., Garatti, S., and Prandini, M. (2009).
\newblock The scenario approach for systems and control design.
\newblock \emph{Annual Reviews in Control}, 33(2), 149--157.

\bibitem[{Coppola et~al.(2022)Coppola, Peruffo, and Mazo~Jr}]{coppola2022data}
Coppola, R., Peruffo, A., and Mazo~Jr, M. (2022).
\newblock Data-driven abstractions for verification of deterministic systems.
\newblock \emph{arXiv:2211.01793}.

\bibitem[{Huang and Wang(2022)}]{huang2022applications}
Huang, B. and Wang, J. (2022).
\newblock Applications of physics-informed neural networks in power systems-a
  review.
\newblock \emph{IEEE Transactions on Power Systems}, 38(1), 572--588.

\bibitem[{Knuth et~al.(2021)Knuth, Chou, Ozay, and
  Berenson}]{knuth2021planning}
Knuth, C., Chou, G., Ozay, N., and Berenson, D. (2021).
\newblock Planning with learned dynamics: Probabilistic guarantees on safety
  and reachability via lipschitz constants.
\newblock \emph{IEEE Robotics and Automation Letters}, 6(3), 5129--5136.

\bibitem[{Krstic and Kokotovic(1995)}]{krstic1995lean}
Krstic, M. and Kokotovic, P.V. (1995).
\newblock Lean backstepping design for a jet engine compressor model.
\newblock In \emph{Proceedings of International Conference on Control
  Applications}, 1047--1052. IEEE.

\bibitem[{Lavaei and Frazzoli(2022)}]{lavaei2022compositional}
Lavaei, A. and Frazzoli, E. (2022).
\newblock Compositional controller synthesis for interconnected stochastic
  systems with {M}arkovian switching.
\newblock In \emph{2022 American Control Conference (ACC)}, 4838--4843.

\bibitem[{Lavaei and Frazzoli(2024)}]{lavaeiTAC2024}
Lavaei, A. and Frazzoli, E. (2024).
\newblock Scalable synthesis of safety barrier certificates for networks of
  stochastic switched systems.
\newblock \emph{IEEE Transactions on Automatic Control}, 69(11), 7294--7309.

\bibitem[{Liu et~al.(2025)Liu, Fitzsimmons, Zhou, and Meng}]{liu2025formally}
Liu, J., Fitzsimmons, M., Zhou, R., and Meng, Y. (2025).
\newblock Formally verified physics-informed neural control {L}yapunov
  functions.
\newblock In \emph{2025 American Control Conference (ACC)}, 1347--1354. IEEE.

\bibitem[{Makdesi et~al.(2021)Makdesi, Girard, and
  Fribourg}]{makdesi2021efficient}
Makdesi, A., Girard, A., and Fribourg, L. (2021).
\newblock Efficient data-driven abstraction of monotone systems with
  disturbances.
\newblock \emph{IFAC-PapersOnLine}, 54(5), 49--54.

\bibitem[{Margellos et~al.(2014)Margellos, Goulart, and
  Lygeros}]{margellos2014road}
Margellos, K., Goulart, P., and Lygeros, J. (2014).
\newblock On the road between robust optimization and the scenario approach for
  chance constrained optimization problems.
\newblock \emph{IEEE Transactions on Automatic Control}, 59(8), 2258--2263.

\bibitem[{McGregor et~al.(2017)McGregor, Gluch, and
  Feiler}]{mcgregor2017analysis}
McGregor, J.D., Gluch, D.P., and Feiler, P.H. (2017).
\newblock Analysis and design of safety-critical, cyber-physical systems.
\newblock \emph{ACM SIGAda Ada Letters}, 36(2), 31--38.

\bibitem[{Nejati et~al.(2023)Nejati, Lavaei, Jagtap, Soudjani, and
  Zamani}]{nejati2023formal}
Nejati, A., Lavaei, A., Jagtap, P., Soudjani, S., and Zamani, M. (2023).
\newblock Formal verification of unknown discrete-and continuous-time systems:
  A data-driven approach.
\newblock \emph{IEEE Transactions on Automatic Control}, 68(5), 3011--3024.

\bibitem[{Nejati et~al.(2024)Nejati, Prakash~Nayak, and
  Schmuck}]{nejati2024context}
Nejati, A., Prakash~Nayak, S., and Schmuck, A.K. (2024).
\newblock Context-triggered games for reactive synthesis over stochastic
  systems via control barrier certificates.
\newblock In \emph{Proceedings of the 27th ACM International Conference on
  Hybrid Systems: Computation and Control}, 1--12.

\bibitem[{Nejati and Zamani(2023)}]{nejati2023data}
Nejati, A. and Zamani, M. (2023).
\newblock Data-driven synthesis of safety controllers via multiple control
  barrier certificates.
\newblock \emph{IEEE Control Systems Letters}.

\bibitem[{Niknejad and Modares(2023)}]{niknejad2023physics}
Niknejad, N. and Modares, H. (2023).
\newblock Physics-informed data-driven safe and optimal control design.
\newblock \emph{IEEE Control Systems Letters}, 8, 285--290.

\bibitem[{Prajna and Jadbabaie(2004)}]{prajna2004safety}
Prajna, S. and Jadbabaie, A. (2004).
\newblock Safety verification of hybrid systems using barrier certificates.
\newblock In \emph{International Workshop on Hybrid Systems: Computation and
  Control}, 477--492.

\bibitem[{Samari et~al.(2026)Samari, Nejati, and Lavaei}]{samari2026data}
Samari, B., Nejati, A., and Lavaei, A. (2026).
\newblock Data-driven control of large-scale networks with formal guarantees: A
  small-gain-free approach.
\newblock \emph{IEEE Transactions on Automatic Control}.

\bibitem[{Wieland and Allg{\"o}wer(2007)}]{wieland2007constructive}
Wieland, P. and Allg{\"o}wer, F. (2007).
\newblock Constructive safety using control barrier functions.
\newblock \emph{IFAC Proceedings Volumes}, 40(12), 462--467.

\bibitem[{Wood and Zhang(1996)}]{wood1996estimation}
Wood, G.R. and Zhang, B.P. (1996).
\newblock Estimation of the {L}ipschitz constant of a function.
\newblock \emph{Journal of Global Optimization}, 8, 91--103.

\bibitem[{Zaker et~al.(2026{\natexlab{a}})Zaker, Akbarzadeh, Samari, and
  Lavaei}]{zaker2024compositional}
Zaker, M., Akbarzadeh, O., Samari, B., and Lavaei, A. (2026{\natexlab{a}}).
\newblock Compositional design of safety controllers for large-scale stochastic
  hybrid systems.
\newblock \emph{Automatica}.

\bibitem[{Zaker et~al.(2026{\natexlab{b}})Zaker, Nejati, and
  Lavaei}]{zaker2025data}
Zaker, M., Nejati, A., and Lavaei, A. (2026{\natexlab{b}}).
\newblock Data-driven safety certificates of infinite networks with unknown
  models and interconnection topologies.
\newblock \emph{Automatica}.

\bibitem[{Zhang et~al.(2022)Zhang, Yin, and Sheil}]{zhang2022physics}
Zhang, P., Yin, Z.Y., and Sheil, B. (2022).
\newblock A physics-informed data-driven approach for consolidation analysis.
\newblock \emph{G{\'e}otechnique}, 74(7), 620--631.

\end{thebibliography}
                                                                        
\end{document}